\documentclass[11pt]{article}
\usepackage[utf8]{inputenc}
\usepackage{tikz}
\usepackage{bbm}
\usepackage{fullpage}
\usetikzlibrary{calc,positioning,shadows.blur,decorations.pathreplacing} 
\usepackage{float,times}
\usepackage{standalone}
\usepackage{amsmath,amsthm,amssymb,fullpage,color}
\newcommand{\R}[0]{{\ensuremath{\mathbb{R}}}}

\newcommand{\I}{{\mathcal{I}}}
\renewcommand{\P}{{\mathcal{P}}}

\newcommand{\beq}{\begin{equation}}
\newcommand{\eeq}{\end{equation}}

\newcommand{\class}{{\tt cl}}
\newtheorem{theorem}{Theorem}[section]
\newtheorem{definition}[theorem]{Definition}

\newtheorem{lemma}[theorem]{Lemma}

\newtheorem{observation}[theorem]{Observation}
\newtheorem{claim}[theorem]{Claim}

\allowdisplaybreaks
\newcounter{note}[section]

\title{Non-uniform Geometric Set Cover and Scheduling on Multiple Machines}
\author{
Nikhil Bansal\thanks{CWI and TU Eindhoven, the Netherlands.
\texttt{bansal@gmail.com}. Supported by a NWO Vici grant 639.023.812.}
\and
Jatin Batra\thanks{CWI, the Netherlands.
\texttt{jatinbatra50@gmail.com}. Supported by a NWO Vici grant 639.023.812.}}
\date{}
\begin{document}

\maketitle
\begin{abstract}
We consider the following general scheduling problem studied recently by Moseley~\cite{Moseley19}. There are $n$ jobs, all released at time $0$, where job $j$ has size $p_j$ and an associated arbitrary non-decreasing cost function $f_j$ of its completion time.
The goal is to find a schedule on $m$ machines with minimum total cost. 
We give an $O(1)$ approximation for the problem,
improving upon the previous $O(\log \log nP)$ bound ($P$ is the maximum to minimum size ratio), and
resolving the open question in \cite{Moseley19}.

\medskip

We first note that the scheduling problem can be reduced to a clean geometric set cover problem where points on a line with arbitrary demands, must be covered by a minimum cost collection of given intervals with {\em non-uniform} capacity profiles.
Unfortunately, current techniques for such problems based on knapsack cover inequalities and low union complexity, completely
lose the geometric structure in the non-uniform capacity profiles and
incur at least an $\Omega(\log\log P)$ loss.

\medskip

To this end, we consider general covering problems with non-uniform capacities, and give a new method to handle capacities
in a way that completely preserves their geometric structure.
This allows us to use sophisticated geometric ideas in a black-box way to avoid the $\Omega(\log \log P)$ loss in previous approaches. 
In addition to the scheduling problem above, we use this approach to obtain $O(1)$ or inverse Ackermann type bounds 
for several basic capacitated covering problems.
\end{abstract}

\section{Introduction}
Recently, Moseley~\cite{Moseley19} considered the following general scheduling problem on multiple machines, 
that captures a wide class of problems studied previously, and much more.
There are $n$ jobs, all released at time $0$, and there are $m$ identical machines. Each job $j \in [n]$ has a processing requirement or size $p_j$ and an associated arbitrary non-decreasing, non-negative cost function $f_j$.
The goal is to find a schedule that minimizes $\sum_j f_j(c_j)$, where $c_j$ is the completion time of job $j$. 
As $f_j$ can be completely arbitrary and job dependent, this models several well-studied objectives such as arbitrary functions of weighted completion times, weighted tardiness, hard deadlines, their mixtures and so on \cite{BansalP14,ImMP14,Moseley19}.
Here we  allow jobs to be preempted and migrated across machines: a simple reduction from the Partition problem shows that no finite approximation is possible otherwise. 
We refer to this as the General Scheduling Problem  (GSP).

In the simpler single machine setting, the problem was considered by Bansal and Pruhs \cite{BansalP14}, who gave an $O(1)$ approximation, based on a strong LP relaxation based on Knapsack cover (KC) inequalities and viewing the constraints in a geometric way. They also gave an $O(\log \log P)$ approximation, where $P$ is the ratio of the maximum to minimum job size, when jobs have arbitrary release times $r_j$. Building on this formulation and underlying geometric ideas, there have been several improvements and breakthrough results on various classic scheduling problems \cite{BatraGK18,FeigeKL19,ImM17,CheungMSV17,HohnMW18,AntoniadisHMVW17}.

The multiple machine setting however seems much harder and the LP relaxation that works for the single machine case is too weak here. Until the recent work of Moseley \cite{Moseley19}, where he introduced new {\em job-cover inequalities}, no good LP relaxation was known. In fact, we still do not know of any good LP for arbitrary release times, and finding one seems to be a challenging open question.
Based on this new formulation, and using a sophisticated rounding based on quasi-uniform sampling \cite{Varadarajan10,ChanGKS12}, Moseley obtained an $O(\log \log nP)$ approximation for GSP.

\subsection{Our Results}
We give an $O(1)$ approximation for  GSP, answering an open question of Moseley \cite{Moseley19}.
\begin{theorem}
\label{thm:main} There is an $O(1)$ approximation for GSP on multiple machines with migration, when all the jobs have identical release times.
\end{theorem}
Our starting point is an observation that the job-cover inequalities in \cite{Moseley19} can be viewed in a clean geometric way.
This allows us to reduce GSP (up to $O(1)$ factors) to a problem of covering demands on a line by intervals that have triangular or rectangular capacities (see Figure \ref{fig:trc}).
Formally, consider the following {\em TRC problem}. 
The input consists of points on a line, and a collection of intervals. A point $p$ has demand $d_p$, and an interval $z=[a_z,b_z]$ has cost $w_z$, and it contributes capacity $c_z(p)$ to $p\in z$. Let us call $c_z$ the capacity profile or simply {\em profile} of $z$. The profiles are either (i) rectangular, i.e.~$c_z(p) = c$ for all $p\in z$, or 
(ii) triangular with slope $1$, i.e.~$c_z(p) = p-a_z$  or $c_z(p) = b_z-p$.
 Find a minimum cost subset $Z$ of intervals so that each demand is satisfied, i.e.~$\sum_{z \in Z} c_{z}(p) \geq d_p$ for all $p$.

In other words, the TRC problem is the same as the well-studied UFP-cover problem on a line~\cite{BarNoyBFNS01,ChakrabartyGK10,   HohnMW18}, except that the capacity profiles of intervals can be {\em non-uniform}, and in particular, triangular with slope $1$. 

We show the following reduction, which implies that to prove Theorem \ref{thm:main}, it suffices to obtain an $O(1)$ approximation for TRC. 
\begin{theorem}
\label{thm:redn}
For any $\alpha \geq 1$, an $\alpha$ approximation for TRC implies a $12\alpha$ approximation for GSP.  
\end{theorem}

\vspace{-3mm}

\paragraph{Capacitated set cover.}
This leads us to consider the TRC problem, and more general covering problems with non-uniform capacities.

A very general and systematic approach for capacitated covering problems,  based on KC inequalities \cite{CarrFLP00}, was developed by Chakrabarty et al.~\cite{ChakrabartyGK10}.
They show that any capacitated covering instance $I$ can be reduced to multiple (uncapacitated) set cover sub-instances, where (roughly speaking) each resulting sub-instance corresponds to a different capacity scale in $I$. If sets in $I$ have uniform capacities, i.e.~$c_z(p)=c_z$ for all $p\in z$, then each set lies in a different sub-instance and hence an $\alpha$-approximate solution to each of the sub-instances can be combined to obtain an $O(\alpha)$-approximation for $I$. 

Unfortunately, this framework gives much worse results when the set capacities are non-uniform, as a set $S$ can now lie in multiple sub-instances, and combining these sub-instances can lead to up to an $O(\log C)$ loss in general (where $C$ is the maximum to minimum capacity ratio). 
In the case of TRC,
one can additionally exploit the geometric structure of profiles (using non-trivial ideas from quasi-uniform sampling) to reduce this loss to $O(\log \log C)$, but this seems to be the limit of this approach. This was also the implicit reason for the $\log \log nP$ loss in the result in \cite{Moseley19}.
It is useful to contrast this with UFP-cover on a line with uniform capacities, for which there are several different $O(1)$ approximations \cite{BarNoyBFNS01, ChakrabartyGK10, HohnMW18}.

\vspace{-3mm}

\paragraph{Handling non-uniform capacities.}
Our main contribution is a new general way to reduce non-uniform capacitated covering problems to (uncapacitated) set cover problems, that preserves the structure of capacities in the original instance. This allows us to use sophisticated geometric machinery based on low union/shallow-cell complexity and small $\epsilon$-nets in a black box way, to avoid the losses inherent in previous approaches.
In particular, this directly gives an $O(1)$ approximation for TRC, and other non-trivial approximation results for much more general problems.

We now state our general result.
Even though this reduction is combinatorial, it is best stated in geometric terms.
\begin{theorem}\label{thm:geomcap}
Consider an instance $\I$ of covering demands $d_p$ of points $p \in \R^{d}$ with a minimum cost collection of sets $z$ with capacity profiles $c_z(p)$.
With each set $z$, associate an induced object in $\R^{d+1}$ given by $\cup_p (p \times [0,c_z(p)])$. 

Then, there is an efficiently constructible set cover instance $\P$ in $\R^{d+1}$ with sets corresponding to the objects induced by the profiles $z$, that satisfies the following property: any $\gamma$-approximation for $\P$ based on rounding the standard LP relaxation for set cover, gives a $9\gamma$-approximation for $\I$.
\end{theorem}
Intuitively, this means that we can simply include the capacity profile of a set into its shape, and work with these new induced (uncapacitated) sets.
For example, in the TRC problem, the induced objects simply becomes rectangles and triangles. See for example, the right half of Figure~\ref{fig:trc}, which shows the induced objects corresponding to the sets for the instance on the left. More formally, for the profile $c_z(p) =c_z$ for $p\in [a_z,b_z]$, the induced object is an axis-parallel rectangle with corners $(a_z,0)$ and $(b_z,c)$.
Similarly if $c_z(p)=b_z-p$ for $p \in [a_z,b_z]$, 
the induced object is a (right angled isosceles) triangle with vertices $(a_z,0), (a_z,b_z-a_z)$ and $(b_z,0)$.

\begin{figure}
    \centering
    \includegraphics[width=\textwidth]{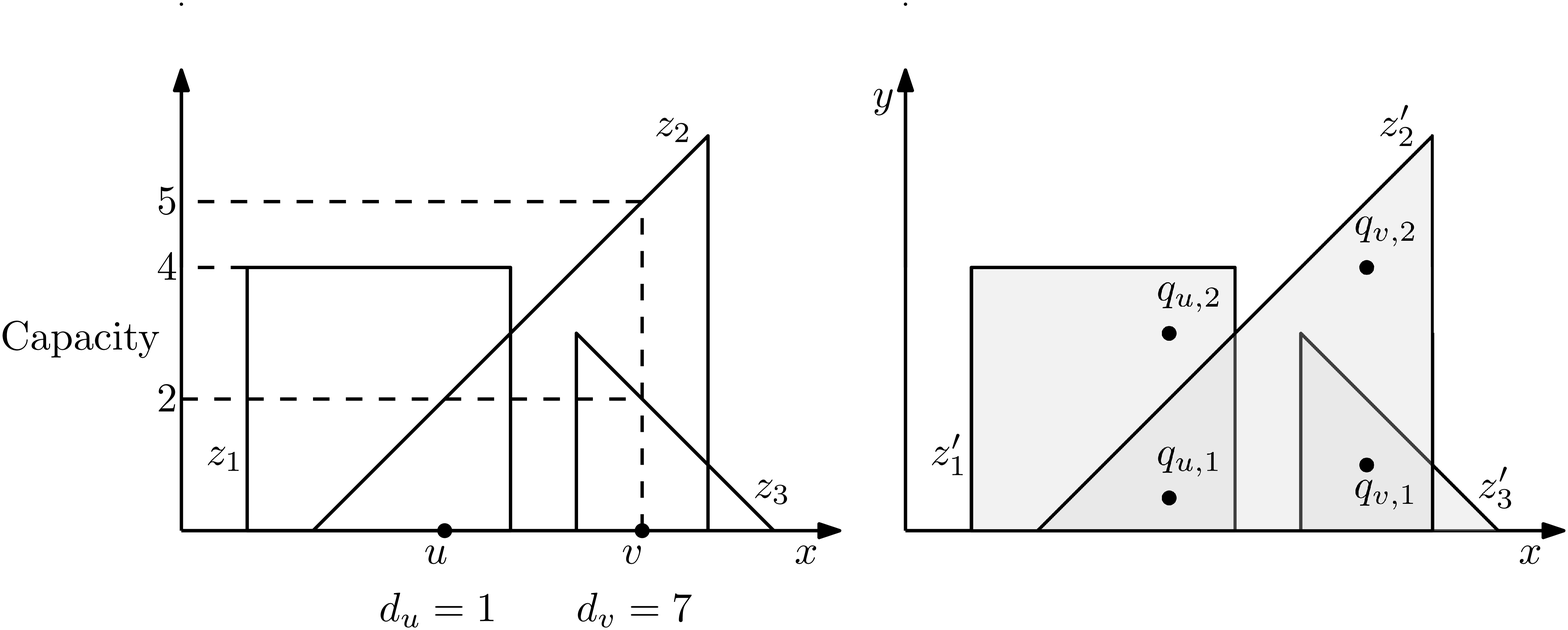}
    \caption{1. A TRC instance with $Z = \{z_1,z_2,z_3\}$ and two points $u,v$. The solution $\{z_2,z_3\}$ is feasible since $z_2(u) + z_3(u) = 2 + 0 \geq 1 = d_{u}$ and $z_2(v) + z_3(v) = 5 + 2 = 7 = d_{v}$.
    \newline 2. The (uncapacitated) set cover problem resulting from the TRC instance. The points to be covered are $q_{u,i},q_{v,i}$ for $i=1,2$ and the sets are $z'_1,z'_2,z'_3$.
    }\label{fig:trc}
\end{figure}

As the union complexity (see definitions below) of any $t$ such triangles is easily checked to be $O(t)$ , the known results  for geometric set cover with low union complexity \cite{Varadarajan10,ChanGKS12,BansalP16} directly give the following result, which together with Theorem \ref{thm:redn} gives Theorem \ref{thm:main}.
\begin{theorem}\label{thm:trc}
There is a $O(1)$-approximation for the TRC problem\footnote{ 
In fact one does not even require that the triangles arising in TRC have slope $1$}.
\end{theorem}

More generally, there has been extensive work on geometric set cover where the underlying sets (objects) have low union complexity, or more generally shallow-cell complexity, and for unweighted geometric set cover where the dual of the underlying set system has small $\epsilon$-nets (we give the relevant definitions in Section \ref{sec:prel}. 
In particular, Theorem \ref{thm:geomcap} together with the results of Varadarajan \cite{Varadarajan10,ChanGKS12, BansalP16} for weighted set cover with low union complexity objects
gives the following result.
\begin{theorem}\label{thm:geomcapeasy}
Consider an instance $\I$ of covering demands $d_p$ of $n$ points $p \in \R^{d}$ with a minimum cost collection of profiles (or sets) $z$ having capacities $c_z(p)$.
If the union complexity of any $t$ of the induced objects in $\R^{d+1}$  is $O(t \phi(t))$, then $\I$ has a $O(\log \phi(n))$ approximation.
\end{theorem}
This also directly extends to the more general shallow-cell complexity framework of Chan et al.~\cite{ChanGKS12} in a straightforward way.
Similarly, for unweighted or minimum cardinality set cover problems, the framework
of Br\"{o}nnimann and Goodrich~\cite{BronnimannG95} and Even et al.~\cite{EvenRS05} gives the following.
\begin{theorem}\label{thm:geomcapeasy2}
Consider an instance $\I$ of covering demands $d_p$ of $n$ points $p \in \R^{d}$ with a minimum cardinality collection of profiles (or sets) $z$ having capacities $c_z(p)$.
If the dual system of the induced objects has $\epsilon$-nets of size $(1/\epsilon) \Phi(1/\epsilon)$, then there is an LP-based
 $O(\Phi(OPT))$ approximation for $\I$.
\end{theorem}
There is a rich history of results proving near linear union complexity for many natural geometric set systems including half-spaces in $\R^3$ and $\R^2$, fat objects in $\R^2$~\cite{AronovBES14}, pseudo-disks in the plane with few intersections~\cite{KedemLPS86} and so on, leading to obtain better guarantees for a variety of covering problems.
 Similarly, many natural geometric set systems admit small $\epsilon$-nets. 
 In general, as any set system of VC-dimension $d$ admits $\epsilon$-nets of size $O((d/\epsilon)\log(1/\epsilon))$~\cite{HausslerW87,BlumerEHW89}, this immediately gives a $O(d\log OPT)$ approximation for such systems. More importantly, much smaller $\epsilon$-nets are known for many natural geometric set systems such as planar disks and pseudo-disks~\cite{PyrgaR08}, axis-parallel rectangles and fat triangles in $\R^2$~\cite{AronovES10}, and axis-parallel boxes, axis-aligned octants and unit cubes in $\R^3$ ~\cite{AronovES10,BoissonnatSTY98}; giving improved guarantees for such systems for covering problems~\cite{BronnimannG95,EvenRS05,clarksonV07,InamdarV18}.

\noindent{\bf Some consequences.} 
Even though Theorems \ref{thm:geomcapeasy} and \ref{thm:geomcapeasy2} follow easily from Theorem \ref{thm:geomcap}, some of the consequences are quite surprising. Below,
in Table~\ref{tab:my_label} we list some of these  non-trivial consequences.
Some of these corollaries also use some additional properties satisfied by the reduction in Theorem \ref{thm:geomcap}, such as that we can split an object into $O(1)$ objects to simplify the geometry while losing only an $O(1)$ approximation. In some cases, such as $s$-piece-wise linear profiles, we can even split an object into $s$ pieces while only losing an $O(\log s)$ factor, due to an interesting property of union complexity \cite{InamdarC19}. Our approach also subsumes the framework of Chakrabarty et al.~\cite{ChakrabartyGK10} for uniform capacity profiles and just as in~\cite{ChakrabartyGK10}, the slack obtained from KC inequalities can be used to round the capacities to integer powers of 2. We discuss these consequences in more detail in Section \ref{sec:appl}.

\begin{table}[hbtp!]
    \centering
    \begin{tabular}{||c|c|c||}
        \hline
         Type of capacity profile
         & Universe
         & Approximation ratio
         \\
         \hline\hline
         Piece-wise linear functions with $s$ pieces
         & $\R$
         & $O(\log s)$
         \\
         \hline
         Polynomial curves of degree at most $s\geq 2$
         & $\R$
         &
         \vtop{\hbox{\strut $O((\alpha(n))^{k-1}) \text{ for } s=2k$}\hbox{\strut $O((\alpha(n))^{k-1}\log \alpha(n))\text{ for } s = 2k+1$}}
         \\
         \hline
         Linear functions
         & $\R^2$
         & $O(1)$\\
         \hline
    \end{tabular}
    \caption{Some applications of our result, here $\alpha(n)$ is the inverse of Ackermann's function}
    \label{tab:my_label}
\end{table}

In general, given the ubiquity of capacitated covering problems in resource allocation and scheduling, and extensive work on geometric problems, we expect that our connection will have more applications.

\subsection{Preliminaries}
\label{sec:prel}
We now describe some notations  and basic results that we will need. 

A crucial distinction will be between set-cover and {\em capacitated} set cover. By set cover we will always mean an instance where the capacity of every set $z$ is exactly $1$.
An element $p$ has some covering requirement $m_p \geq 1$, which means that it should be covered with at least $m_p$ distinct sets to be satisfied (strictly speaking, this should be called multi-cover, but we use set cover to avoid clutter with too many variants of set cover, and as far as we know all results that hold for set cover also hold for multi-cover.)

In a capacitated set cover problem, the elements have arbitrary demands $d_p$ and sets have capacity profile (or function) $c_z(\cdot)$. We call the problem {\em non-uniform} if $c_z(p)$ for $p\in z$ can depend on $p$, and {\em uniform} if $c_z(p)=c_z$ for all $p \in z$.  To make the distinction even clearer, note that we use the notation $d_p$ and $m_p$ to distinguish the demands of points in capacitated and uncapacitated versions.

\noindent {\bf Knapsack-cover (KC) inequalities.}
Capacities can make a covering problem much harder. Already for a single element, we get the NP-hard Knapsack Cover problem (given items with cost $w_i$ and size $c_i$, find a minimum cost subset of items to cover a knapsack of size $d$). The following natural LP relaxation for it 
\[ \min \sum_i x_i w_i \qquad \text{s.t.} \quad  \sum_i c_i x_i \geq d \qquad    x_i \in [0,1], \quad \forall i \in [n]\]
 has an arbitrarily large integrality gap.
The following stronger LP based on exponentially many KC inequalities, was introduced by Carr et al.,~\cite{CarrFLP00} and they showed that this reduces  the integrality gap to $2$.
\[ \min \sum_i x_i w_i \quad \text{s.t. } \quad  \sum_{ i \notin S} \min( p_i, d-p(S))\,x_i   \geq   d - p(S)  \quad \forall S \subset [n],\,p(S) \leq d \quad    x_i \in [0,1], \quad \forall i \in [n]\]
where  $p(S) = \sum_{i \in S} p_i$.
Roughly, the inequalities say that even if all the items in $S$ are chosen, 
the residual demand of $d-p(S)$ must still be covered by items not in $S$. Even though exponentially large, the LP can be solved to any accuracy in polynomial time. These inequalities have been very
useful for various capacitated covering problems.  An alternate perspective using primal-dual and local ratio methods is in \cite{CarnesS15, BarNoyBFNS01}.

\vspace{-3mm}

\paragraph{Geometric set cover.} We will be interested in settings where the sets correspond to geometric objects in $\R^d$, in fixed dimension $d$, such as rectangles, triangles, boxes and so on. 
Given a collection $X$ of such geometric objects, the {\em union complexity} of $X$ is the number of faces of all dimensions (vertices, edges and so on) on the boundary of the object formed by the union of all objects in $X$. We say that a geometric set system has union complexity function $\phi(\cdot)$, if for every $t$, every collection of $t$ sets has union complexity at most $t \phi(t)$.

In a breakthrough result \cite{Varadarajan10}, Varadarajan developed a powerful quasi-uniform sampling technique to give improved bounds for geometric set cover problems with low union complexity. This was extended to more general shallow-cell complexity \cite{ChanGKS12}\footnote{We do not discuss this here to keep the exposition simple, but Theorem \ref{thm:geomcapeasy} also works directly in this setting.} and to multi-cover~\cite{BansalP16}. These results give the following.
\begin{theorem}%
\label{thm:unionmulticover}
There is a polynomial time LP-based $O(\log \phi(n))$ approximation for instances of geometric  set cover on $n$ points, where the set system has union complexity function $\phi(\cdot)$.
\end{theorem}
For several natural geometric objects in $2$ and $3$ dimensions (axis-parallel rectangles, fat triangles, disks etc.), $\phi(t)$ is typically $O(1)$ or $O(\log t)$~\cite{Varadarajan10}, leading to better $O(1)$ or $O(\log \log n)$ approximations.

\noindent {\bf{$\epsilon$-nets}} defined below have been very useful for understanding the complexity of geometric set systems. 
\begin{definition}\label{def:epsnet}
For a set system $X$, $Y \subseteq X$ is said to be an \emph{$\epsilon$-net} if $Y$ covers all elements which are covered by at least an $\epsilon$-fraction of the sets in $X$ i.e. if $e \in Y$ for all elements $e$ satisfying $|\{X:e\in X\}| \geq \epsilon |X|$. 
\end{definition}
Bronnimann and Goodrich~\cite{BronnimannG95} showed (theorem~\ref{thm:epsnet}) that a set system $(U,F)$ admits good approximations for unweighted set cover if the dual system has small $\epsilon$-nets. Here the dual system to $(U,F)$ is the system $(U^*,F^*)$ where $U^* = F$ and $F^*$ consists of the sets $\{S \in F: e \in S\}$ for each $e \in U$.

\begin{theorem}\label{thm:epsnet}
Suppose for a set system $(U,F)$, an $\epsilon$-net of size $(1/\epsilon) \Phi(1/\epsilon)$ can be found for the dual system in polynomial time for any $\epsilon>0$\footnote{Techincally, this requires an $\epsilon$-net with respect to weights of the sets, but this is not an issue in natural applications.}. Then the unweighted set cover problem for $(U,F)$ admits an LP-based $O(\Phi(OPT))$ approximation.
\end{theorem}

\subsection{Organization and Overview}
The rest of the paper is organized as follows.
In Section~\ref{sec:gsp}, we consider the GSP problem and show that it reduces to the TRC problem.
In Section \ref{sec:geomcap}, we give the reduction from non-uniform capacitated geometric set cover to geometric set cover, and prove our main result Theorem~\ref{thm:geomcap}.
In Section~\ref{subsec:trc}, we show how Theorem~\ref{thm:trc} for TRC follows from Theorem \ref{thm:geomcapeasy}, 
and in Section~\ref{sec:appl} we give other applications of our results beyond TRC, and describe some additional useful properties of the reduction in Theorem \ref{thm:geomcap}.

\section{Geometric view of GSP}\label{sec:gsp}
We now consider the GSP problem. The jobs are indexed $1,\ldots,n$, and all are released at time $0$.
Without loss of generality, we assume that all job sizes $p_j$ are integers, and time is slotted, referring to the time interval $(i-1,i]$ as time $i$. As $p_j$ are integers, we can assume that at most one job executes at any slot on any machine, and that the functions $f_j$ are defined on non-negative integers.

As the objective $\sum_j f_j(c_j)$ only depends on the completion times $c_j$, solving GSP is equivalent to finding \emph{deadlines} $c_j$, so that each job $j$ can be feasibly scheduled by $c_j$, and $\sum_j f_j(c_j)$ is minimized.
By standard preprocessing, see e.g.~\cite{BansalP14}, losing at most factor $2$ in the objective, we can assume that 
$f_j(x)$ is piece-wise constant and takes at most $O(\log n)$ values\footnote{Roughly, we can round $f_j(x)$ to powers of $2$ or down to $0$ if $f_j(x) \leq \text{Opt}/n$, where $\text{Opt}$ is some guessed upper bound on the optimum value.}). 
So for each $j$, it suffices to consider $k=O(\log n)$ candidate deadlines $c_{j,0},\ldots,c_{j,k}$. Here, $c_{j,0}$ is the latest time until which $f_j(x)=0$.

Let $v = \sum_j p_j$.
A key result of Moseley was the following characterization of feasible deadlines.
\begin{theorem}[\cite{Moseley19}]
\label{lem:moseley}
Given a set of deadlines $c_j$, a feasible schedule exists iff the following holds: 
\beq 
\label{eq:cond} \sum_{j \in [n]} \min(p_j,\max(c_j-b,0)) \geq \sum_{j \in [n]}  p_j - mb  \qquad \text{for all } b=0,1,\ldots,v.
\eeq
\end{theorem}
This result is based on considering a flow network to find a feasible schedule, and observing that any min-cut must have a specific form. For completeness, we give a proof in the Appendix~\ref{subsec:flow}.

\vspace{-3mm}
\paragraph{The geometric reduction.}
We show that the feasibility condition \eqref{eq:cond} has a clean geometric view. 
We first define a wedge.
For parameters $p,c$, let a wedge $u_{p,c}$ be the function  (see Figure \ref{fig:wedge}). 
\[u_{p,c}(t) = \min(p, c-t) \text{ for $ t \leq c$, and $0$ otherwise.}\]
In other words, $u_{p,c}$ is decreasing at slope $1$ from $t=\min(c-p,0)$ to $t=c$, and is constant with value $\min(p,c)$ from $t=0$ to $t=\min(c-p,0)$. %

A key observation is that \eqref{eq:cond} can be viewed equivalently as follows.
\begin{observation}
Given a candidate deadline $c_j$ for each job $j$, consider the wedge $u_{p_j,c_j}$. 
Consider each point $b=0,1,\ldots,v$, with demand  $d_b = \sum_j p_j - mb$. 
Then the collection of deadlines $c_j$ are feasible iff the wedges satisfy all the demands, i.e.~$\sum_j u_{p_j,c_j}(b) \geq  d_b$ for all  $b$.
\end{observation}
As GSP is equivalent to finding feasible $c_j$'s with minimum total cost,  GSP reduces (without any loss in objective) to the following geometric {\em wedge-cover problem}:
For each job $j$, and every possible deadline $c$ for $j$, there is a wedge $u_{p_j,c}$ of cost $f_j(c)$. 
Choose \emph{exactly one} wedge for each job $j$ so that the demand $d_b$ for each $b$ is satisfied, and the total cost of the chosen wedges is minimized.
This is almost a covering problem, except that exactly {\em one} wedge must be picked for a job, which is a packing condition,. Another issue is that as $v=O(nP)$, there are potentially exponentially many points $b$ to be covered. Both these issues are easily fixed as we show next.
\begin{figure}
    \centering
    \includegraphics[width=\textwidth]{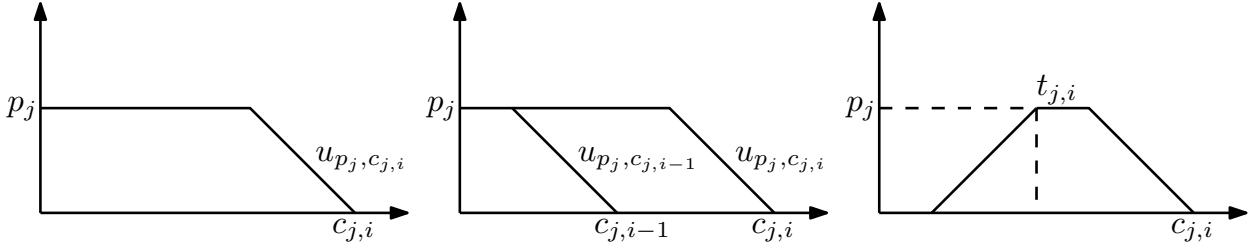}
    \caption{1. A single wedge $W_{j,i}\quad$ 2. Splitting into $W_{j,i} - W_{j,i-1}\quad$ 3. Trapezoid $T_{j,i}$}
    \label{fig:wedge}
\end{figure}

\noindent {\bf Getting a purely covering problem.}
By the preprocessing for $f_j$'s described earlier, each $j$ has only $k=O(\log n)$ possible deadlines $c_{j,0},\ldots,c_{j,k}$, with geometrically increasing costs $f_{j}(c_{j,i}) \geq 2 f_j(c_{j,i-1})$.

Let $W_{j,i}$ denote the wedge $u_{p_j,c_{j,i}}$ for job $j$ with deadline $c_{j,i}$.
For each $i \in [k]$, let us define $T_{j,i} = W_{j,i}- W_{j,i-1}$ and
assign $T_{i,j}$ cost $f_j(c_{j,i})$.
For $i=0$, set $T_{j,0} = W_{j,0}$ and assign it cost $f_j(c_{j,0})=0$.
 Note that $T_{j,i}$ has
a ``trapezoid" profile given by $t_{j,i} =u_{p_j,c_{j,i}} - u_{p_j,c_{j,i-1}}$
(see Figure \ref{fig:wedge}),
with a rectangular part (possibly empty), and up to two triangular parts, either rising or falling at slope $1$.  
Moreover, for any point $z$,  the capacity of $W_{j,i}$ at $z$ is exactly equal to the total capacity of $T_{j,0},\ldots,T_{j,i}$, i.e.~$u_{p_j,c_{j,i}}(z) = \sum_{g=0}^i t_{j,g}(z).$

\noindent{\bf Trapezoid-cover problem.} Given a GSP instance, first preprocess it and create $T_{j,i}$ as defined above. Consider the {\em trapezoid-cover problem} of finding a minimum cost collection of $T_{j,i}$ (possibly choosing several for each job $j$) so that the demand $d_b = \sum_j p_j - mb$ for each $b$ is satisfied. 

Given a GSP instance $I$, let $T$ denote the corresponding trapezoid cover instance obtained by the reduction above. We have the following simple relation between the value of their optimum solutions.
\begin{lemma}
\label{l:cost}  $OPT(I) \leq OPT(T) \leq 4 OPT(I)$.  
\end{lemma}
\begin{proof} Consider some feasible solution to $I$ with value $w$. After preprocessing (rounding the $f_j(c)$), this solution, and hence the corresponding wedge-cover instance has a solution of value at most $2w$. If this solution picks the wedge $W_{j,i}$ for job $j$, pick trapezoids $T_{j,g}$ for $g=0,\ldots,i$ in the trapezoid-cover instance. This gives a feasible cover, and as $t_{j,g}$ have geometrically increasing costs, the total cost of $t_{j,g}$ is at most twice that of $W_{j,i}$. Hence $OPT(T)\leq 4w$.

Conversely, given a feasible solution to $T$ of cost $w$, for each job $j$ let $i(j)$ be the highest index such that $T_{j,i(j)}$ is in the cover. If no trapezoid for $j$ is chosen, set $i(j)=0$. Then choosing the wedge $W_{j,i(j)}$, for each $j$ gives a feasible wedge-cover with cost at most $w$, and hence a feasible solution to GSP.
\end{proof}

\noindent{\bf The TRC problem.} To obtain the problem of covering with rectangles and triangles, split each trapezoid $T_{j,i}$, into at most one rectangle and two triangles, each with the same cost as that of $T_{j,i}$. This loses at most another factor $3$ in the approximation. 
Next a simple idea, whose proof is in Appendix \ref{subsec:flow}, ensures that the TRC instance has size polynomial in $n$ and the bit-complexity of $f_j$ and $p_j$.
\begin{lemma} The number of points $b$ in the TRC instance can be reduced to $O(n \log n)$. 
\end{lemma}
Together with Lemma \ref{l:cost}, this proves Theorem \ref{thm:redn}.

\section{Non-uniform capacitated cover}\label{sec:geomcap}
We now prove Theorem \ref{thm:geomcap}, which we restate here for convenience.

\begin{theorem}
\label{thm:geomcap2}
Consider an instance $\I$ of covering demands $d_p$ of points $p \in \R^{d}$ with a minimum cost collection of profiles (or sets) $z$ with capacities $c_z(p)$.
With each profile $z$, associate an induced object in $\R^{d+1}$ given by $\cup_p (p \times [0,c_z(p)])$. 

Then, there is an efficiently constructible set cover instance $\P$ in $\R^{d+1}$ with sets corresponding to the objects induced by the profiles $z$, that satisfies the following property: any $\gamma$-approximation for $\P$ based on rounding the standard LP relaxation, gives a $9\gamma$-approximation for $\I$.
\end{theorem}

To do this, we first solve an LP for $\I$ strengthened by $KC$ inequalities. Given some solution $x$ to this LP with cost $w^*$, we create the set cover instance $\P$ with sets corresponding to objects induced from $\I$ as described above. The instance $\P$ will satisfy the following two properties. First, any feasible integral solution to $\P$ of cost $w(P)$ gives a solution to $\I$ of cost at most $w(P) + O(1) w^*$. Second, the basic set cover LP for $\P$ has a feasible solution of value at most $O(1) w^*$. It is easily seen that these properties directly imply Theorem \ref{thm:geomcap2}. 
We now give the details.

\noindent {\bf KC LP for $\I$.} 
The standard LP relaxation for $\I$ has variables $x_z$ for each set $z \in Z$, and is the following.
\[ \min \sum_z w_z x_z      \qquad  \sum_z c_z(p) x_z \geq d_p \quad \forall p, \qquad  x_z \in [0,1] \]
This has unbounded integrality gap, and we strengthen the constraint for each $p$ by KC inequalities to get,
\[
\sum_{ z \notin S} \min( c_z(p), d_p -c_S(p))\,x_z   \geq   d_p - c_S(p)  \qquad \forall p, \forall S \subset Z,\,c_S(p) \leq d_p  \]
where $Z$ denotes the collection of sets (objects) in $\I$ and
for a subset $S$ of objects, $c_S(p) = \sum_{z \in S} c_z(p)$.

Let $x$ denote some optimum solution to this LP, and let $w^*$ be its cost. To create instance $\P$, we first pre-process $x$ as follows.
Let $S$ be the set of objects $z$ with $x_z \geq 1/\beta$, where $\beta = O(1)$ will be specified later. 
We select the objects in $S$ integrally in the cover for $\I$. 
Let $d'_p = \max(0,d_p - c_S(p))$ be the residual demand of point $p$. Then by the KC inequalities (applied to this set $S$), $x$ satisfies the following for each $p$,
\[ \sum_{z \in Z \setminus S} \min( c_z(p), d'_p) x_z \geq d'_p.\]
Let $Z'$ denote the residual objects $Z \setminus S$, and note that $x_z \leq 1/\beta$ for each $z \in Z'$. Consider the solution $x'$ where $x'_z= \beta x_z$ for $z \in Z'$ (and hence $x'_z \leq 1$ for $z \in Z'$). Then, the solution $x'$ satisfies
\beq\label{eq:bcover} \min(c_z(p),d'_p) x'_z \geq \beta d'_p  \qquad \forall p,\eeq
covering the residual demands $d'_p$ by an extra $\beta$ factor,
and has cost at most $\beta \sum_{z \in Z'} w_z x_z \leq \beta w^*$.

\vspace{-2mm}

\paragraph{The instance $\P$.} We use $x'$ to create the set cover instance $\P$.
For each $z \in Z'$, we create the induced object $\tilde{z}$ as defined in Theorem \ref{thm:geomcap2}. For each point $p \in I$, we will create multiple points $q_{p,j}$ for each $j=0,1,\ldots,\lfloor \log_2 d'_p \rfloor$, where $q_{p,j} = (p,2^{-j}d'_p)$, i.e.~the first $d$ coordinates of $q_{p,j}$ are the same as that of $p$, and the $d+1$-th coordinate is $2^{-j}d'_p$. 

For $z \in Z'$, we define the class $\class(z,p)$ of $z$ with respect to $p$ as the smallest non-negative integer $j$ such that $c_z(p) \geq 2^{-j} d'_p$. In other words, if $c_z(p) \in [2^{-j} d'_p,2^{-j+1}d'_p)$ for $j \geq 1$, and $0$ if $c_z(p) \geq d'_p$.
We set the covering requirement of $q_{p,j}$ to be $m_{p,j} = \lfloor \sum_{z : \class(z,p)\leq j } x'_{z} \rfloor.$

This completes the description of $\P$. We now show that it satisfies the properties mentioned earlier.

Let $\tilde{x}$ denote an LP solution induced from $x'$  to objects in $\P$ by defining $\tilde{x}_{\tilde{z}} = x'_z$ for each $z \in Z'$. 
\begin{claim} \label{cl:x'} $\tilde{x}$ is a feasible solution for the basic set cover LP for $\P$.
\end{claim}
\begin{proof} We need to show that each $q_{p,j}$ is covered fractionally to extent at least $m_{p,j}$. Now, by construction of $\P$, the point $q_{p,j}$ is covered by some object $\tilde{z}$ iff $c_z(p) \geq 2^{j} d'_p$. By the definition of $m_{p,j}$ and $\class(z,p)$ this directly gives
\[ \sum_{\tilde{z}:q_{p,j} \in \tilde{z}} \tilde{x}_z = \sum_{z} x'_z \cdot  \mathbbm{1}[c_z(p)\geq 2^{-j}d'_p] =\sum_{z:\class(z,p)\leq j} x'_z \geq m_{p,j}. \qedhere\] 
\end{proof}

Let $\tilde{Z}$ denote the set of objects in $\P$. Given a subset $\tilde{S} \subseteq \tilde{Z}$ of objects, let $S'$ denote the corresponding sets in $Z'$.
\begin{claim} 
\label{cl:feasible}
If $\beta \geq 8$, then given any feasible integral set cover $\tilde{S} \subseteq \tilde{Z}$ for $\P$, the corresponding $S'$ satisfies the demands $d'_p$ for each $p$, i.e. $\sum_{z \in S'} c_z(p) \geq d'_p$. 
\end{claim}
\begin{proof} 
Fix a point $p$. For any $j$, let $\tilde{S}_{p,j}$ be the objects in $\tilde{S}$ that cover $q_{p,j}$. As $\tilde{S}$ is feasible, $|\tilde{S}_{p,j}| \geq m_{p,j}$ 
and by construction each $z \in Z'$ corresponding to an object in $\tilde{S}_{p,j}$ has class $\class(z,j) \leq j$. Let $r_{p,j} = |\tilde{S}_{p,j}|-|\tilde{S}_{p,j-1}|$ be the number of objects of class exactly $j$. Here $\tilde{S}_{p,-1}$ is the empty set.

Then the capacity contributed to $p$ by objects in $S'$ satisfies
\begin{align*}
    \sum_{z \in S'} c_z(p) & \geq  \sum_{j\geq 0} r_{p,j} 2^{-j} d'_{p} = \sum_{j\geq 0}(|\tilde{S}_{p,j}|-|\tilde{S}_{p,j-1}|) 2^{-j} d'_p = \sum_{j \geq 0} |\tilde{S}_{p,j}|2^{-j-1} d'_p \\ 
    & \geq  \sum_{j \geq 0} m_{p,j} 2^{-j-1} d'_p \geq \Bigg( \sum_{j \geq 0} \Big(\sum_{z : \class(z,p)\leq j } x'_{z} \Big) -1 \Bigg) 2^{-j-1} d'_p 
\end{align*}
where the first inequality uses that $c_z(p) \geq 2^{-j} d_p'$ if $\class(z,p)=j$, the second inequality uses that $\tilde{S}_{p,j} \geq m_{p,j}$ and the third inequality uses that by definition $m_{p,j} \geq (\sum_{z : \class(z,p)\leq j } x'_{z}) -1 $.

As $\sum_{j \geq 0} 2^{-j-1} d_p' \leq d_p'$,  and lower bounding $\sum_{z : \class(z,p)\leq j } x'_{z}$ by $\sum_{z : \class(z,p)= j } x'_{z}$, the quantity above is at least 
\[ \sum_{j \geq 0} \Big(\sum_{z : \class(z,p) = j } x'_{z} \Big) 2^{-j-1}d'_p   -d'_p \]
As $\min(c_z(p), d_p')  \leq d_p' 2^{-j+1}$ if $\class(z,p)=j$ for any $j$, this is at least
\[ \sum_{j \geq 0} \Big(\sum_{z : \class(z,p) = j } \frac{1}{4} \min(c_z(p),d'_p) x'_z \Big)  -d'_p \geq \frac{\beta}{4}  d'_p - d'_p \geq d'_p\]
where the first inequality uses \eqref{eq:bcover} and the second inequality that $\beta \geq 8$.
\end{proof}

We can now complete the proof of Theorem \ref{thm:geomcap2}.
Let $\tilde{S}\subseteq \tilde{Z}$ be a  $\gamma$-approximate solution for $\P$ with respect to the basic set cover LP relaxation for $\P$.
Then by Claim \eqref{cl:x'} as $\tilde{x}$ is a feasible LP solution, the cost of $\tilde{S}$ is at most $\gamma w(\tilde{x}) = \gamma w(x')$. As $x' \leq \beta x$, this is at most $\gamma \beta w^*$.
Moreover, by Claim \ref{cl:feasible}, the solution $S'$ satisfies the residual demands and hence $S \cup S'$ is a feasible solution for $I$. As $S$ has cost at most $\beta w^*$, this gives an overall approximation guarantee of $(\gamma +1 ) \beta \leq 9\gamma$ for $\beta=8$.

\subsection{$O(1)$ approximation for TRC}\label{subsec:trc}

We now apply Theorem~\ref{thm:geomcap} to obtain an $O(1)$ approximation for TRC. 
The TRC instance has three types of objects: axis-aligned rectangles and right angled triangles of slope 1 and -1 respectively (denote these objects by $R$,$T(1)$ and $T(-1)$ respectively). 
By scaling an optimum solution to the LP relaxation by a factor of 3, it is enough to obtain an LP-based $O(1)$ approximation for each one of the classes $R,T(1),T(-1)$. It suffices to upper bound the union complexity by $O(t)$ due to Theorem~\ref{thm:unionmulticover}.

To upper bound the union complexity, consider the upper envelope of $t$ induced objects, and define an \emph{edge} to be a maximal connected piece of the upper envelope belonging to the same object (see figure~\ref{fig:edges}). Then clearly,
the union complexity of $t$ induced objects is $O(t + \text{number of edges})$.
\begin{figure}
    \centering
    \includegraphics[width=\textwidth]{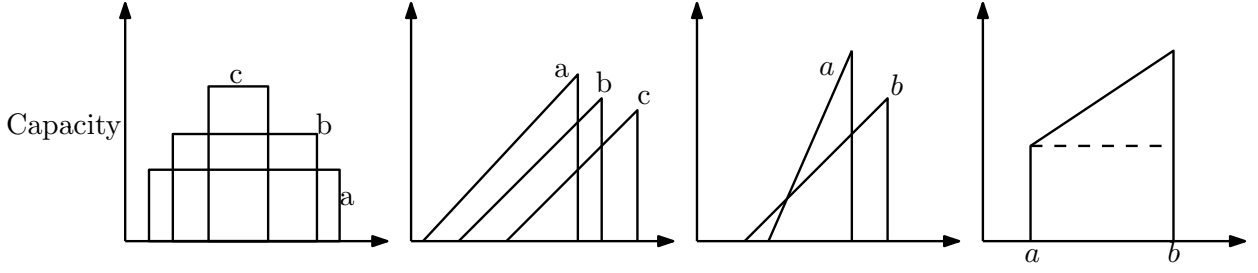}
    \caption{1. Edges of upper envelope of by rectangles $a,b,c$ have sequence $abcba$ \newline 2. Edges of upper envelope of right triangles $a,b,c$ of slope $1$ have sequence $abc$\newline 3. Upper envelopes of right triangles $a,b$ of positive slope intersect in $\leq 2$ points, with edge sequence $aba$.\newline
    4. Decomposing a linear profile supported on $[a,b]$ into a rectangular and right triangular profile.
    }
    \label{fig:edges}
\end{figure}

To bound the number of edges, let us label them by the objects they belong to and consider the sequence of cells from left to right in the upper envelope of $t$ objects. To bound the length of the sequence, let us use the notion of a Davenport-Schinzel sequence~\cite{AgarwalS10}, as we will also need this later for other applications.

\begin{definition}\label{def:ds}
A Davenport-Schinzel (DS) sequence of order $s$ on $n$ symbols is a sequence $a_1,a_2,\ldots,a_m$ satisfying $a_i \neq a_{i+1}$ for all $1 \leq i \leq m-1$, which has no alternating subsequence of length $s+2$ i.e. no subsequence like $aba\cdots a$. Let $\lambda_s(n)$ denote the maximum length of a DS sequence of order $s$ on $n$ symbols.
\end{definition}

It is well-known (and easily verified) that for any two rectangles $\{a,b\}$ in $R$, there cannot be a subsequence of edges $abab$, and hence the edge sequence is $DS$ of order $2$. Similarly, there cannot be a subsequence of edges $aba$ for any one of $T(1)$ or $T(-1)$ (as the slopes are $1$), and hence the edge sequence is $DS$ of order $1$ (see Figure \ref{fig:edges}).
We will consider triangles with arbitrary slopes (or more generally piece-wise linear functions) in Section~\ref{sec:appl}.

Using the well-known bounds $\lambda_1(n) = n, \lambda_2(n) \leq 2n-1$~\cite{AgarwalS10}, we have that the union complexity of $t$ objects in any one of $R,T(1)$ or $T(-1)$ is $O(t)$, implying our $O(1)$-approximation for TRC.

\subsection{Further applications of Theorem~\ref{thm:geomcap}}\label{sec:appl}
We now consider various other applications of Theorem~\ref{thm:geomcap}. Some of these applications use additional flexibility that our setting and the proof of Theorem~\ref{thm:geomcap} gives such as being able to round capacities to powers of $2$, or splitting an object into multiple objects.

\vspace{-3mm}

 \paragraph{Decomposing capacity profiles.} Instead of using Theorem~\ref{thm:geomcap} directly, note that we can also first decompose a capacity profile $z$ into two (or more) simpler capacity profiles $z_1,z_2$ with $c_z = c_{z_1}+c_{z_2}$, assigning them the same cost as $z$. This operation at most doubles the approximation factor, but allows us to work with simpler geometric structures giving improved guarantees, similar to our algorithm for TRC.
 
 We collect below some applications (Table~\ref{tab:my_label}) using these ideas.

\vspace{-3mm}

\paragraph{Covering points on the plane with linear capacity profiles.} Consider an instance $\I$ of covering points $(x,y)$ in the plane with sets $z$ with capacity profiles $c_z(x,y) = (ax+by)_{+}$ where $m_{+}$ denotes $\max(m,0)$. Then, applying the framework of theorem~\ref{thm:geomcap} to $\I$,
we get an instance $\P$ of covering points in the upper half-space ($(x,y,q)$ with $q \geq 0$) of $\R^3$, with the induced objects corresponding to half-spaces $\{(x,y,q): 0 \leq q\leq ax+by\}$ for capacity profile $(ax+by)_{+}$. Now, a point $(x,y,q)$ with $q\geq 0$ lies in the object induced by $(ax+by)_{+}$ if and only if $(x,y,q)$ lies in the half-space $ax+by\geq q$. Hence, $\P$ is simply a problem of covering points by half-spaces in $\R^3$. It is well-known that the union complexity of $t$ half-spaces in $\R^3$ is $O(t)$, giving a $O(1)$ approximation.
To see this, note first that the union of half-planes is the complement of a convex polyhedron. Now use point-plane duality for the planes involved in the half-planes to note the equivalence of the faces of the convex polyhedron and those of the convex hull of the dual points. Then, use McMullen's Upper Bound theorem~\cite{Matousek02} to bound the number of faces of the convex hull.

\vspace{-3mm}

\paragraph{Covering points on a line.}
Let us fix the universe to be the real line. 
Let us note that if we have a collection of capacity profiles such that any two of them intersect in at most $s$ positions, then by definition of DS sequences, the upper envelope of the induced objects obtained by applying Theorem~\ref{thm:geomcap2} is a DS sequence of order $s$, and hence has union complexity $O(\lambda_s(n))$.
Given the well-known and non-trivial bounds   for $\lambda_s(n)$~\cite{AgarwalS10} (see table~\ref{tab:ds}),  Theorem \ref{thm:geomcapeasy} gives several interesting results.

\begin{table}[h!]
    \centering
    \begin{tabular}{|c|c|}
        \hline
         Order $s$ & $\lambda_s(n)$\\
         \hline\hline
         1 & $n$ \\
         \hline
         2 & $2n-1$ \\
         \hline
         3 & $O(n \alpha(n))$\\
         \hline
         4 & $O(n 2^{\alpha(n)})$\\
         \hline
         $> 4$, $s=2k$ & $O(n 2^{O((\alpha(n))^{k-1})})$\\
         \hline
         $> 4$, $s=2k+1$ & $O(n2^{O((\alpha(n))^{k-1} \log \alpha(n))})$\\
         \hline
    \end{tabular}
    \caption{Maximum length of Davenport-Schinzel sequence $\lambda_s(n)$ for various $s$, here $\alpha(n)$ is the very slowly growing inverse of Ackermann's function}
    \label{tab:ds}
\end{table}

\vspace{-4mm}

\paragraph{Linear capacity profiles on intervals.}
Consider linear capacity profiles supported on intervals $[a,b]$ i.e. profiles $z$ of the form $c_z(x) =  mx+q \text{ for } x \in [a,b]$ (this generalizes the triangular profiles of slope $1$ considered before). Losing at most factor $2$, assume that $m>0$ (we partition into two sub-problems each with slopes $m$ of the same sign). At another factor 2 loss, break the linear profile into two profiles $c_{z_1}(x) = m(x-a)$ (a \emph{right triangular profile}) and $c_{z_2}(x) = ma+q$ (a \emph{rectangular} profile) both supported on $[a,b]$. Applying Theorem~\ref{thm:geomcap} we obtain a problem of covering points in the plane by induced objects that are right triangles (of positive slope) and rectangles, both touching the $x$-axis. As 
any two such rectangles intersect in at most $2$ places, and similarly for any two such triangles,  the union complexity will be $O( \lambda_2(t)) = O(t)$, giving an $O(1)$-approximation.

\vspace{-3mm}

\paragraph{Piece-wise linear profiles.}
We can extend this to piece-wise linear profiles with at most $s$ pieces. Naively decomposing the piece-wise linear profile into $s$ linear pieces would give an $O(s)$-approximation. Instead, think of the object induced by a $s$-piece-wise linear profile as a union of two objects - one is the union of $s$ right triangular profiles and the other is a union of $s$ rectangular profiles.

Chekuri and Inamdar~\cite{InamdarC19} showed the following result for geometric set systems.
\begin{theorem}\label{cl:inamdar}~\cite{InamdarC19}
Suppose the union complexity of any $t$ objects in a geometric set family $F$ in $\R^d$ is $O(t \phi(t))$. Then, there is an LP-based $O(\log s + \log \phi(ts))$-approximation for covering by objects formed by taking unions of $\leq s$ objects in $F$.
\end{theorem}
Theorem~\ref{cl:inamdar} thus gives $O(\log s)$ approximations for both the right triangular profile unions and the rectangular profile unions, giving an $O(\log s)$ approximation for covering by $s$-piece-wise linear functions.

\vspace{-3mm}

\paragraph{Low degree polynomial curves.}
Suppose each profile is a degree $\leq s$ polynomial curve. Then, the number of intersections of the envelopes of any pair of induced objects is at most $s$ and hence the union complexity will be $O(\lambda_{s}(t))$. Theorem \ref{thm:geomcapeasy} now gives the results for degree $s$ polynomials given in Table~\ref{tab:my_label}.

\noindent{\bf Uniform capacities.} For uniform capacity profiles, Theorem~\ref{thm:geomcap} directly generalizes the framework of Chakrabarty et al.~\cite{ChakrabartyGK10}. In particular, the notion of priority cover which they introduce is exactly captured by the point-set incidences of the induced objects in our construction. Similarly as in~\cite{ChakrabartyGK10}, we can use the slack obtained from  KC inequalities to assume that there are only $O(\log C)$ different capacities where $C$ is the range of the capacities.
This reduced number of capacities could be useful for many applications (see e.g.~\cite{BansalP14}). For instance if the union complexity of the underlying set system is $t\phi(t)$, then the union complexity of the induced objects will only be $O(t \log C \phi(t))$ giving an $O(\log \log C +\log \phi(t))$ approximation.

{\small
\bibliography{biblio}  }

\appendix

\section{Deadline Feasibility and Flow Network}\label{subsec:flow}
Here we give a necessary and sufficient condition for feasibility of candidate deadlines, as developed by Moseley~\cite{Moseley19}. Let $c_j$ be the candidate deadline for job $j$. Consider the following flow network.

\noindent {\bf The flow network $G$.}  See figure \ref{fig:flownetwork}.
There is a source node $s$ and sink node $t$.
Layer $1$ consists of $n$ nodes, one for each job $j$.
Layer $2$ consists of at most $v= \sum_j p_j$ nodes, one for each possible time slot where a job can execute.
There is a directed edge from $s$ to $j$ with capacity $p_j$. Each job node $j$ has a directed edge of capacity $1$ to each time slot in $\{1,\ldots,c_j\}$. From each time slot in $[v]$, there is a directed edge to $t$ with capacity $m$. 

It is clear that a feasible schedule exists if and only if a flow of value 
$\sum_j p_j$ exists (in fact any feasible integral flow corresponds to a valid schedule, and conversely). In other words, if and only if the minimum $s$-$t$ cut does not have value less than $\sum_j p_j$.  
Note that while the network is of size $O(nP)$ which can be exponentially large, we do not actually solve it algorithmically, and only use it to derive valid inequalities.

\begin{figure}[H]
\centering
\includegraphics[width=0.25\textwidth]{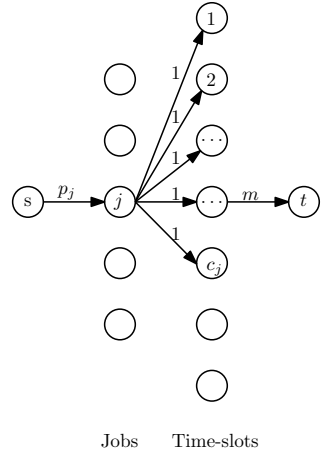}
\caption{The flow network for testing feasibility}\label{fig:flownetwork}
\end{figure}

\begin{lemma}[\cite{Moseley19}]
\label{lem:moseley1}
Given a set of completion times $c_j$, there is a feasible schedule if and only if the following condition holds for $b=0,1,\ldots,v$.
\beq 
\label{eq:cond2} \sum_{j \in [n]} \min(p_j,\max(c_j-b,0)) \geq \sum_{j \in [n]}  p_j - mb
\eeq
\end{lemma}
\begin{proof}
For a subset $J \subseteq [n]$ of jobs and subset $T \subseteq [v]$ of times slots, let $(J,T)$ denote the $s$-$t$ cut with the part containing $s$ as $\{s\} \cup J \cup T$, and let   $\delta(J,T)$ denote its value. By considering the contribution of each type of edge to $\delta(J,T)$, we have that
\beq \delta(J,T) = \sum_{j \notin J} p_j +  \sum_{j \in J} | [c_j]\setminus T| +  m |T|, 
\label{eq:cut-val}
\eeq 
where $[c_j]\setminus T$ denotes the time slots in $\{1,\ldots,c_j\}$ that do not lie in $T$.

For any cut $(J,T)$, note that replacing some $t \in T$ with an earlier time $t'< t, t' \notin T$ can only reduce $\delta(J,T)$. Indeed, for $T'=T \cup\{t'\} \setminus \{t\}$, $|T|=|T'|$ and for each $j$,  $| [c_j]\setminus T'| \leq | [c_j]\setminus T|$ as $t'<t$, and hence by \eqref{eq:cut-val}, $\delta(J,T') \leq \delta(J,T)$.
Repeating this process, we can assume that there is some min-cut of the form  $\delta(J,[b])$  (possibly with $b=0$, corresponding to $T = \emptyset$). By \eqref{eq:cut-val}, we get 
\[ \delta(J,[b]) =  \sum_{j \notin J} p_j +  \sum_{j \in J} \max(c_j-b,0) +  m b  \]
Finally, for a fixed $b$, note that each $j$ contributes exactly $\max(c_j-b,0)$ or $p_j$ or depending on whether $j \in J$ or not.
This implies that 
\[ \min_{J \subseteq [n]}   \delta(J,[b]) =  \sum_{j \in [n]} \min(p_j ,\max(c_j-b,0)) +  m b,  \]
and hence the min-cut is at least $\sum_j p_j$ iff the right side is at least $ \sum_j p_j$ for all $b \geq 0$, giving the result.
\end{proof}

\begin{lemma} The number of points $b$ in the TRC can be reduced to $O(n \log n)$. 
\end{lemma}
In our TRC instance, we have one point for each $b=0,1,2,\ldots,v$ and at most $O(n \log n)$ sets (rectangles, triangles) $Z$. Identify each set $z \in Z$ by the interval $[a_z,b_z]$ on which it is supported. These intervals divide the line into at most $2 |Z| + 1$ sub-intervals (equivalence classes of points lying in the same set of intervals in $Z$). We show that for each sub-interval, it suffices to have two points in the instance.
\begin{claim}
Consider any sub-interval $U$ above and let $u_1$ and $u_2$ be its left-most and right-most points. If a subset $Z' \subseteq Z$ satisfies the demands of $u_1,u_2$, then $Z'$ satisfies the demands of all points in $U$.
\end{claim}
\begin{proof}
Observe that for any collection of sets $Z'$, the total capacity $c_{Z'}(u)$ for points $u$ in $U$ is a linear function as no new object begins or ends strictly inside $U$, and the capacity of any object $z$ is a linear function in its interval $[a_z,b_z]$. If $Z'$ satisfies the demands of $u_1,u_2$, then 
\beq c_{Z'}(u_1) \geq \sum_j p_j - m \cdot u_1 \qquad \text{and } \qquad 
\label{eq:l2} c_{Z'}(u_2) \geq \sum_j p_j - m \cdot u_2 \eeq
Any point $u \in U$ can be expressed as a convex combination $u = \lambda_1 u_1 + \lambda_2 u_2$ with $\lambda_1,\lambda_2 \geq 0, \lambda_1+\lambda_2 = 1$. Multiplying inequalities above by $\lambda_1$ and $\lambda_2$ respectively and adding, and using the linearity of $c_{Z'}$ gives
\[ c_Z(u) = c_{Z'}(\lambda_1 u_1 + \lambda_2 u_2) = \lambda_1 c_{Z'}(u_1) + \lambda_2 c_{Z'}(u_2) \geq \sum_j p_j - m (\lambda_1 u_1 + \lambda_2 u_2) = \sum_j p_j - m u. \qedhere \]
\end{proof}

\end{document}